\documentclass{article}

\date{Dec 2020}

\usepackage{microtype}
\usepackage{graphicx}
\usepackage{subfigure}
\usepackage{booktabs} 
\RequirePackage{amsfonts}
\RequirePackage{textcomp}
\RequirePackage{theorem}
\RequirePackage{graphicx}
\RequirePackage{amsmath}
\RequirePackage{url}

\newtheorem{lemma}{\textsc{Lemma}}[section]
\newtheorem{corollary}{\textsc{Corollary}}[section]
\newtheorem{theorem}{\textsc{Theorem}}[section]

\newtheorem{remark}{\textsc{Remark}}[section]
\newtheorem{prop}{\textsc{Proposition}}[section]

\newtheorem{example}{\textsc{Example}}

\DeclareRobustCommand{\qed}{\hfill$\square$}
\newenvironment{proof}[1][]{\noindent\emph{Proof #1-- }}{\qed\newline}

\DeclareRobustCommand{\calA}[0]{{\cal A}}

\DeclareRobustCommand{\calF}[0]{{\cal F}}

\DeclareRobustCommand{\E}[0]{\mathbb{E}}

\DeclareRobustCommand{\P}[0]{\mathbb{P}}
\DeclareRobustCommand{\Q}[0]{\mathbb{Q}}
\DeclareRobustCommand{\R}[0]{\mathbb{R}}

\DeclareRobustCommand{\supstack}[2]{\stackrel{\mbox{{\scriptsize ${#1}$}}}{ {#2} }}     

\DeclareRobustCommand{\half}{\frac12}

\DeclareRobustCommand{\IND}{\mbox{\textsf{\sl 1}}}
\DeclareRobustCommand{\1}[1][ ]{\ensuremath{\IND_{#1}\,}}

\DeclareRobustCommand{\Lin}[1]{\mathrm{lin}\ }

\DeclareRobustCommand{\prob}[1][P]{{\ensuremath{\mathbb{#1}}}}

\DeclareRobustCommand{\EX}[2][\E]{\ensuremath{{#1}\left[\, {#2}\, \right]}}

\DeclareRobustCommand{\Var}[1][ ]{\ensuremath{\mathrm{Var[#1]}}}

\DeclareRobustCommand{\essinf}[0]{\mathrm{ess inf}}   
  
\DeclareRobustCommand{\sig}{\ensuremath{\sigma}}           

\DeclareRobustCommand{\eqlabel}[1]{\label{eq:#1}}                                     
\DeclareRobustCommand{\eq}[1]{\begin{equation}\eqlabel{#1}}                           
\DeclareRobustCommand{\eqx}{\begin{equation}}                                         
\DeclareRobustCommand{\eqend}{\end{equation}}                                         
\DeclareRobustCommand{\eqref}[1]{\mbox{(\ref{eq:#1})}}                                

\DeclareRobustCommand{\eqary}{\begin{eqnarray*}}
\DeclareRobustCommand{\eqaryend}{\end{eqnarray*}}

\DeclareRobustCommand{\aup}{{\hat{\epsilon}}}
\DeclareRobustCommand{\adn}{{\breve{\epsilon}}}

\DeclareRobustCommand{\gmup}{{\hat{\gamma}}}
\DeclareRobustCommand{\gmdn}{{\breve{\gamma}}}

\usepackage{hyperref}

\usepackage[accepted]{icml2019}

\icmltitlerunning{Learning Risk-Neutral Implied Volatility Dynamics}
\begin{document}

\icmltitle{Deep Hedging: Learning Risk-Neutral Implied Volatility Dynamics}
\icmldate{March 2021. This version: July 2021}



\icmlsetsymbol{equal}{*}

\begin{icmlauthorlist}
\icmlauthor{Hans Buehler}{equal,jpm}
\icmlauthor{Phillip Murray}{equal,jpm,imp}
\icmlauthor{Mikko S. Pakkanen}{equal,imp}
\icmlauthor{Ben Wood}{jpm}
\end{icmlauthorlist}

\icmlaffiliation{jpm}{AAO Equities QR, JP Morgan, London}
\icmlaffiliation{imp}{Imperial College London}

\icmlcorrespondingauthor{Phillip Murray}{phillip.murray@jpmorgan.com}

\icmlkeywords{Deep Hedging, Complete Market, Stochastic Volatility}

\vskip 0.3in
\setlength{\baselineskip}{12pt}



\printAffiliationsAndNotice{\icmlEqualContribution} 

\begin{abstract}
We present a numerically efficient approach for learning a risk-neutral measure for paths of simulated spot and 	option prices up to a finite horizon
 under convex transaction costs and convex trading constraints. 

This approach can then  be used to implement a \emph{stochastic implied volatility} model in the following two steps:
\begin{enumerate}
\item Train a market simulator for option prices, as discussed for example in our recent work~\citet{DHOPT};
\item Find a risk-neutral density, specifically the minimal entropy martingale measure.
\end{enumerate}
The resulting model can be used for risk-neutral pricing, or for Deep Hedging~\cite{DH} in the case of transaction costs or trading constraints.

To motivate the proposed approach, we also show that market dynamics are free from ``statistical arbitrage" in the absence of transaction
costs if and only if they follow a risk-neutral measure. We additionally provide a more general characterization in the presence of convex
transaction costs and trading constraints. 

These results can be seen as an analogue of the fundamental theorem of asset pricing for statistical arbitrage under trading frictions and are of independent interest.

\end{abstract}


\section{Introduction}

One of the long-standing challenges of quantitiative finance is the development of a tractable 
\emph{stochastic implied volatility model}. 

The aim of such a model is the simulation of spot and a number of options across strikes and maturities under risk neutral  
dynamics. Naturally, we will want to express option prices in a reasonably intuitive parametrization such as
Black~\&~Scholes implied volatilities. 
Moreover, we typically look at floating option surfaces, where implied volatilities are parameterized in time-to-maturity and moneyness, not
in fixed maturity dates or cash strikes. 
Figure~\ref{fig:histSPX} shows examples of such historical surfaces.

\begin{figure}[h]
  \centering
  \includegraphics[width=0.8\textwidth]{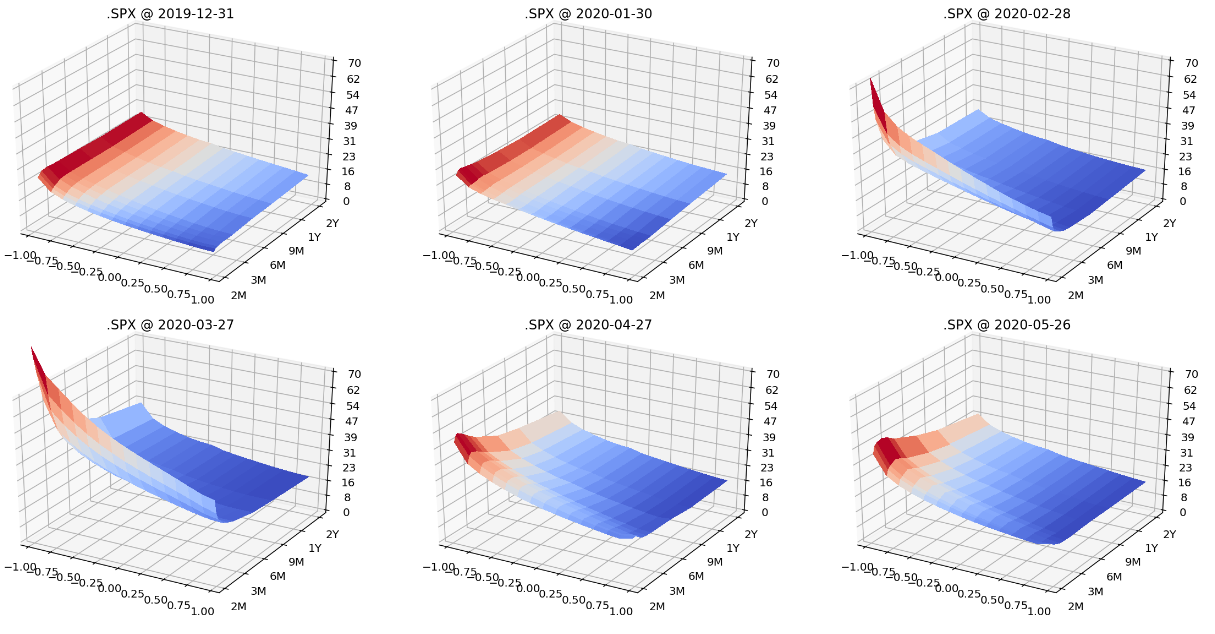}
  \caption{Historical implied volatility surfaces for S\&P500, in delta strikes}
  \label{fig:histSPX}
\end{figure}

The primary use case in practice for such models is in risk-neutral pricing where the value of any portfolio of financial instruments is given 
as an expectation. Since this is a linear operation it lends itself to mass-parallel evaluation of large portfolios of instruments. Computing a consistent hedge requires merely the computation of first order derivatives with respect to each of the hedging instruments.\footnote{It is worth noting that the vast majority of 
financial derivatives in equities and FX depends on spot, and not usually on option prices. Key exceptions are futures and options on volatilty indices
such as VIX and VSTOXX.} This approach, however, does not incorporate transaction costs or trading restrictions such as liquidity constraints or risk limits.
For such a setup we have introduced our \emph{Deep Hedging}~\citet{DH} approach. As we will explain below even this non-linear approach to 
derivative pricing and hedging is easier to deploy in production when run
under a risk-neutral measure.

\subsubsection*{Stochastic Implied Volatility}

The classic quantitative finance approach to stochastic implied volatility models is specifying the model dynamics in a suitable
 parameter space under the risk-neutral measure with model dynamics designed to be somewhat analytically tractable.
 The working paper~\citet{BRACE} summarizes various Heath-Jarrow-Morton (HJM) type drift conditions for diffusion dynamics option surfaces.
 It is clear that if the model generates paths with \emph{static arbitrage} such as negative prices for butterflies or
calendar spreads, no equivalent risk-neutral measure can exist. It is therefore imperative to use a parameterization of the option price surface which allows efficient control for absence of arbitrage.
This is in particular not the case for Black \& Scholes implied volatilities.
 
First applicable results for a term structure of implied volatilities were presented in~\citet{SCHOEN}, and extended to the entire 
option price surface in the seminal work of~\citet{WISSEL}. There, option prices are parameterized in what are called ``implied local volatilities" and which we will refer to as \emph{discrete local volatilities}
because, they are a discrete form of Dupire's local volatilities. 
Non-negativity alone ensures absence of static 
arbitrage; translation from discrete local volatilities
to option prices is numerically efficient with modern machine learning tools.\footnote{The complexity of translating a discrete local volatility 
surface to option prices is equivalent to solving the respective very sparse implicit finite different
scheme, inverting a tridiagonal matrix in every step. This operation is natively supported by TensorFlow.} 
Wissel then proceeds to describe the required continuous time drift adjustment for a diffusion driving a grid of such discrete local volatilities as a function of the free parameters.
Unnaturally, in his approach the resulting spot diffusion takes only discrete values at the strikes of the options at each maturity date and the approach is limited to a set grid of options defined in cash strikes
and fixed maturities.

More recently, a number of works
have shown that when representing an option surface with a L\'{e}vy kernel we can derive suitable HJM conditions on the parameters of the diffusion of the L\'{e}vy kernel such that the resulting stock
price is arbitrage-free, c.f.~\citet{HJMLEVY} and the references therein. Simulation of the respective model requires solving the respective Fourier equations for the spot price and options at each step in the path.

While impressive from an analytical perspective, each of these approaches faces challenges in practical implementation as the model dynamics are made to fit a particular analytical outcome, and are inherently
continuous-time models. They are therefore not 
in a natural form conducive to statistical training.

Here, we present an alternative approach which is primarily driven by the desire to have realistic implied volatilty dynamics for a finite grid of options
with efficient numerics, rather than tractable explicit analytics: we wish to first build a market simulator
of the market under the statistical measure, and then \emph{learn} an associated
risk-neutral measure using non-parametric machine learning methods. While we prefer to represent option prices with discrete local volatility for numerical efficiency, our approach will work with any other parametrization which
is free of static arbitrage.

Moreover, we postulate that trading incurs transaction costs and is subject to trading restrictions such as liquidity constraints and risk limits. We therefore
aim to extend our results to find measures under which a trader cannot make money within her trading constraints and after paying trading costs. While such a measure is not necessarily a martingale
measure, it ensures that the expected returns from trading any instruments are within the bid/ask spread. Such a measure
is a natural candidate for running our previously discussed Deep Hedging~\cite{DH} approach below.

\subsubsection*{Machine Learning Risk-Neutral Volatility Dynamics}

When training a machine learning model on historical market data, it will pick up the historical drift present in actual financial instruments.
For example, when looking at the S\&P500 index and its options from 2015 to 2021, the model may infer that being long SPY calls is 
a winning strategy. More subtly, it makes economic sense that puts, in particular, trade at a risk premium, i.e.~that selling puts
will on average generate profits. 
These effects are not statistically incorrect, but place too much trust in the implicitly estimated magnitude and persistence of the returns of those instruments.
For the purpose of risk managing a portfolio of financial instruments we therefore deem it preferable to ``remove the drift" of 
any tradable instrument.
This is a complex operation for a surface of options, and requires the construction
of an equivalent martingale measure.\footnote{Since we are in discrete time with continuous variables, 
there is no sensible notion of a ``unique" martingale
measure.}

We propose the following approach: given a market simulator we first train a ``Deep Hedging" model under the statistical measure $\P$ 
 to find an optimal \emph{statistical arbitrage} strategy $a^*$, i.e.~a strategy
which starting from an empty portfolio makes the most out of the perceived drift opportunities in the market by trading across all instruments
considering prevailing transaction costs and trading constraints. If the strategy exists and is finite, we use it to construct a change of measure through
	\eqx
		\frac{ d\Q^* }{ d\P } = \frac{ e^{ -G(a^*) } }{ \E_{\P} [ e^{- G(a^*)  } ] }. \ 
	\eqend
	
Under this measure, no strategy starting from an empty portfolio
can have a positive risk-adjusted return after adjusting for cost. Furthermore, if there are no  transaction costs  or trading constraints,
then $\Q^*$ is the \emph{minimal entropy martingale measure}.
We will elucidate this novel insight in Theorems ~\ref{thm:nsa1} and~\ref{th:robustmemm}, which are of independent interest.

As in our previous work~\citet{DH}, the approach is entirely ``model free" in the sense that the numerical implementation
of finding the risk-neutral measure change does not depend on the simulator. This allows efficient division of work between 
experts in machine learning to
train a simulator with realistic volatility dynamics, and experts in classic quantitative finance who focus on efficient implementation of the 
risk-neutral expectation machinery.

In particular, while not discussed here, our approach also lends itself to non-equity markets and
multiple assets across several currencies.

\subsubsection*{Arbitrage Free Parametrization of Option Surfaces}

It is clear that if the market simulator generates paths with \emph{static arbitrage} such as negative prices for butterflies or
calendar spreads, no such measure can exist. To avoid static arbitrage, we
 propose to use discrete local volatilities in our market simulators as a numerically efficient sparse parametrization
of our option surfaces.
Figure \ref{fig:histSPXdlv} shows such historic discrete local volatility surfaces as illustration.

We have shown in
our previous work~\citet{DHOPT} how to use modern machine methods such as generative adversarial networks to train realistic 
simulators
of the discrete local volatility surface.  
Here, we will refer to a simpler vector-autoregression model
which performs less well when generating paths over several times steps, but is easily interpretable.

We want to point out that our 
approach to computing the risk-neutral measure for a market simulator does not depend on the choice of using
discrete local volatilities.

\subsubsection*{Deep Hedging with Risk-Neutral Dynamics}

The original motivation for the work in this article is to provide a means to ``remove the drift" when applying the
non-linear Deep Hedging approach in the presence of transaction costs and trading constraints. There, we value portfolios of financial instruments via indifference pricing by minimizing a convex risk measure.
As we have noted in our original work~\citet{DH}, presence of statistical arbitrage means that the
returned strategy is a mixture of the sought, ``true" hedging strategy and proprietary trading strategy which seeks to make profits. 

Running Deep Hedging under a risk-neutral measure removes this concern. In spirit this is the same approach as risk-managing
a portfolio of stocks and futures using an industry ``covariance" risk model that also does not provide a view on 
the drifts of the underlying names. We comment on the theoretical interpretation of this approach below.

\subsection*{Related Work}

There have been a number of previous works on the statistics and simulation of implied volatility surfaces. A~PCA approach to the dynamics of the surface 
was presented in~\citet{CONT}. Simulators that are constrained to prevent static arbitrage have been less well researched. 
Besides our own work~\citet{DHOPT}, we also want to point at~\citet{IMPVAE}, where the authors build a variational auto encoder to simulate the full implied volatility surface, including the ability
to penalize for absence of arbitrage somewhat. While they report no issues in application for scenario analysis, their approach does not strictly prevent arbitrage as the no-arbitrage condition is only introduced through soft penalities, but is not
an inherent property of the resulting volatility surfaces.
 Our approach applies to this method, too, provided the resulting surfaces are indeed arbitrage free.

We not aware of any application
of machine learning 
methods for estimating the measure change towards the minimal entropy martingale measure for paths of simulated derivatives; 
the closest related field of work is that of estimating a stochastic ``pricing kernel"; see~\citet{COCH}
and the numerous references therein. In this approach, observed option data
is fit statistically to historical samples of observed spot data. This amounts in spirit to learning directly the density of
the risk-neutral measure under the statistical measure for the spot price process itself. As the number of possible densities is large, an additional penalty function
is required. If this penalty is the relative entropy, then in the absence of transaction costs
the results here are analytically equivalent when applied to just the asset spot price. We are not aware of any work
applying this approach to markets of derivatives as primary tradable assets as opposed to just spot prices.  

\section{Theoretical Background}

In this section we will focus on the theoretical underpinning of our approach. We will use a notation which we aim to be accessible to both readers
with stochastic calculus and reinforcement learning background. Throughout, we assume discounting rates, funding, dividends, and repo rates are
zero. Extension to the case where they are non-zero and deterministic is straightforward.\footnote{For the stochastic rates case
care need to be taken on the choice of discounting and investment of future proceeds.}

We will trade over time steps $0=t_0<\cdots<t_m=T$ where $T$ is the maximum maturity. For each $t \in \{ t_0,\ldots,t_m\}$, we denote by 
$s_t$ the \textbf{state} of the market at  time $t$, including relevant information from the past. 
The state represents all information
available to us, and therefore specifically also the information required to compute the mid-prices $H_t^{(t)}=(H_t^{(t,1)},\ldots,H_t^{(t,n)})$ of the hedging instruments
available to trade at $t$. In other words we may assume that $H^{(t)}_t \equiv H^{(t)}_t(s_t)$ is a function of $s_t$. Mathematically speaking,
the sequence of states $(s_t)_{t=0,\ldots,T}$ generates a sequence $(\calF_t)_{t=0,\ldots,T}$ of $\sig$-algebras forming a filtration. Being generative means
that any $\calF_t$-measurable function $f(\cdot)$ can be written as a function of $s_t$ as $f \equiv f(\cdot;s_t)$. We will generally say ``measurable"
when a variable  with index~$t$ is a function of the state $s_t$.

We further assume that for each instrument we observe at time $T$ a final mark-to-market mid-value $H_T^{(t,i)}$ which will usually be the sum of any cashflows along the path, and
which is also assumed to be a function of $s_T$. That means $s_T$ must contain sufficient information from the past along the path:
for example, if an instrument tradable at $t$ is a call with relative strike $k_i$
and time-to-maturity
$\tau_i \leq T-t$ on a spot price process $S_t$, then $H_T^{(t,i)} = (S_{t+\tau_i}/S_t - k_i)^+$. 

We also assume that $H_t \in L^\infty$, where $L^\infty$ is the set of essentially bounded measurable random variables;\footnote{See also Remark~\ref{rem:bounded}.} furthermore we simplify notation by stipulating that
the total number of instruments at each timestep is always $n$.

At each time step $t$ we may chose an \textbf{action} $a_t$ to trade in the hedging instruments~$H_t^{(t)}$ based on the information available in the state $s_t$, i.e.~$a_t\equiv a_t(s_t)$.
 Given $s_t$, each $a_t$ is constrained to a convex set $\mathcal{A}_t \equiv \mathcal{A}_t(s_t)$. We assume $0\in \mathcal{A}_t$.  It
 defines the set of admissible actions and represents risk and liquidity constraints. For an action $a_t$ also define $a^\pm_t := \max(0, \pm a_t)$ elementwise.

 For the current discussion we will assume that $\calA_t$ has a specific structure which will allow us to make structural statements such as Theorem~\ref{thm:nsa1}
 on page~\pageref{thm:nsa1}. Assume we observe a state $s_t$. We postulate then that if there is \emph{any} admissible strategy $a$ with $a^{(i)}_t>0$ ($a^{(i)}_t<0$), then there is a $\aup^{(i)}_t>0$
 ($\adn^{(i)}_t>0$)
 such that the strategy of doing nothing except buying $\aup^{(i)}_t$ (selling $\adn^{(i)}_t$)  units of $H^{(i,t)}$ is also admissible. 

We model \textbf{trading costs} and \textbf{trading restrictions} via a non-negative measurable \textbf{generalized cost function} $c_t(a_t)\equiv c_t(a_t;s_t)$ with values in $[0,\infty]$, which is convex in $a_t$ and
normalized to $c_t(0)=0$. We further make the
 structural assumption that 
 $\lim_{\varepsilon \downarrow 0} \frac1\varepsilon 
c_t(\varepsilon a_t^i) = \gmup_t^i \cdot a_t^i{}^+ + \gmdn_t^i \cdot a^i_t{}^-$ for a non-negative measurable ask spread $\gmup_t^i$ and bid spread $\gmdn_t^i$ with values in $[0,\infty]^n$.\footnote{We note that the use of this property below is invariant with respect to a linear transformation, i.e.~the statements can be generalized to linear combinations of tradable instruments.}

The convex set $\calA_t(s_t) := \{ x: c_t(x;s_t) < \infty\}$ defines the set of \textbf{admissible actions} available
at time $t$, conditional on $s_t$. We call our problem \emph{unconstrained} if $c$ is finite for all for real-valued actions in which case~$\calA_t \equiv \R^n$ for all $t$ and $s_t$.
We also set $\calA := \calA_0 \times \cdots \times \calA_{m-1}$ which denotes the set of \textbf{admissible policies}.

\begin{remark}
	We note that the following construction leads to an acceptable generalized cost function: assume that $\tilde \calA_t(s_t)$ is a convex set given $s_t$, and that $\tilde c(a_t;s_t)$
	is convex in $a_t$ and normalized, but finite and defined only on $\tilde A_t$. Then, $ c_t(a_t;s_t) := \tilde c_t(a_t) \1_{a_t \in \calA_t} + \infty \1_{a_t \not\in \calA_t}$ is a valid cost function.
	
A~common use case are sets of the form $\tilde \calA_t = \bigcap_k \{ a_t:\ g^k( a_t; s_t ) \leq u^k(s_t) \}$ where $g^k$ are convex
functions and where $u^k(s)\geq 0$. For example, $ a_t \cdot \mathrm{Vega}_t \leq \mathrm{MaxVega}_t$ for classic ``Vega" constraints or
 $ a_t \Sigma_t a_t \leq \mathrm{MaxRisk}_t$ for quadratic risk constraints.
\end{remark}
 
Convexity implies $c_t(a_t) \geq \gmup_t \cdot a^+_t + \gmdn_t \cdot a^-_t \geq 0$
 for any admissible policy $a$. 
We note that the $i$th asset can be bought (for a finite price) only on $\{\gmup^i_t<\infty\}$ and sold (for a finite price) only on $\{\gmdn^i_t<\infty\}$.
 In the following, we will refer to the joint vector with $\gamma_t=(\gmup_t,\gmdn_t)$ to ease notation.

We say that transaction costs are \emph{proportional} if simply $c_t(a_t) \equiv \gmup_t \cdot a^+_t + \gmdn_t \cdot a^-_t$ over $\calA_t$ and infinite elsewhere.

 The terminal \textbf{gain} of implementing the \textbf{trading policy} $a = (a_0,\ldots,a_{m-1}) \in \calA := \calA_0 \times \cdots \times \calA_{m-1}$ is then given by
 \eq{Ga}
 	G(a) := \sum_{t=0}^{m-1} a_t  \cdot ( H^{(t)}_T - H^{(t)}_t ) - c_t(a_t) \ 
 \eqend
Note that if $a\in\calA$ and transaction costs are proportional then $\lambda a\in\calA$ with $G(\lambda a) =\lambda G(a)$ for $\lambda\in[0,1]$.

We make the further assumption that $a \in L^\infty$ for any admissible policy $a\in\calA$, which implies in particular that $\E[|G(a)|]<\infty$.  \footnote{This assumption is to some degree necessary. Our results may not hold if the
 gains of admissible policies have no finite moments. See Remark \ref{rem:bounded} for a concrete example.}
   
\begin{remark}
\label{dh}
Our slightly unusual notation of taking the performance of each instrument to maturity reflects our ambition to look at option simulators for ``floating" implied volatility surfaces where
the observed financial instruments change from step to step. If the simulator were to simulate the same options 
with fixed cash strikes and maturities across
time steps, with prices $H_t$ of the same options available at every time step,
then the usual notation applies:
 \eq{gain_old}
 	G(a) = \sum_{t=0}^{m-1} \delta_t\, dH_t - c_t(a_t) \ , \ \ \ \delta_t := \delta_{t-1} + a_t \ , \ \ \ \delta_{-1} = 0 \ .
 \eqend
 This was the notation used in the original Deep Hedging paper of ~\citet{DH}.
\end{remark}   

\begin{remark}
	We note that \eqref{Ga} implies that spot and options which mature after $T$ are valued at mid-prices. The above can easily be extended to take into account liquidation cost at maturity.
\end{remark}

In order to assess the performance of a trading strategy, we need a risk-adjusted measure of performance. Reversing sign relative to~\citet{DH}
we use a family of normalized \emph{monetary utility functions} $U_\lambda$ parameterized by $\lambda\in(0,\infty)$, where $U_\lambda$ maps from 
all random variables integrable under $\P$ to $[-\infty,\infty)$. 

In this article we focus on the (certainty
equivalent of the) \emph{entropy},
\eq{entropy}
	U_\lambda(X) := -\frac1\lambda \log \E\left[ e^{-\lambda X} \right] \ .
\eqend
If $X$ is normally distributed then the entropy reduces to the
well-known mean-variance metric $U_\lambda(X) = \E[X] - \half \lambda \Var[X]$ pioneered by~\citet{MV}.
It is well known that $U_\lambda$ is monotone\footnote{Lemma~\ref{lem:g-prop} establishes ths condition for the entropy.}, concave and cash-invariant.

\begin{remark}\label{rem:bounded}
We stress that $U_0(X) := \lim_{\lambda\downarrow 0}U_\lambda(X) = \E[X]$ is not generally satisfied by the entropy. A~classic example is a variable $X=1-\exp(Y-\half)$ where
$Y$ is standard normal. For all $\lambda>0$ we have $U_\lambda(X)=-\infty$ while $\E_\P[X] = 0$. This is why we have here insistented on essentially bounded $H$ and $a$.\footnote{C.f.~the proof for Propostion ~\ref{prop:alllambda} on page~\pageref{prop:alllambda}.}
\end{remark}

\subsection{Removing Statistical Arbitrage}

A natural question is whether there is a strategy $a$ which has positive risk-adjusted expected return $U_\lambda ( G(a)) > 0$.
We call such a strategy a \emph{statistical arbitrage strategy}. As we have discussed in the introduction, this is not an unusual 
situation: practical strategies such as selling puts systematically \emph{are} on average profitable if we have sufficient risk capacity
and can withstand the occasional large loss. \\
Define 
\eq{g}
	g_\lambda := \sup_a \ U_\lambda\big(\, G(a) \, \big) \ .
\eqend
This optimisation program can be implemented efficiently with modern ``reinforcement learning" policy search using AAD packages such
as TensorFlow, c.f.~\citet{DH}. We note that $g_0=0$ means that even 
a risk neutral trader cannot find profitable opportunities in the market. We then say that the market is \emph{free
from statistical arbitrage}. This definition is justified as absence of statistical arbitrage then implies
also $g_\lambda = 0$ for all $\lambda\geq0$, as a result of the following lemma.

\begin{lemma}\label{lem:g-prop}
The map $\lambda \mapsto g_\lambda$ for $\lambda \in [0,\infty]$ is non-increasing and non-negative.
\end{lemma}

In the specific case of the entropy, we have the following stronger result, proved in the Appendix on page~\pageref{proof:alllambda}:
 \begin{prop}\label{prop:alllambda}
	Assume $U_\lambda$ is the entropy, and that transaction costs are proportional.
	Then, $g_\lambda > 0$ for some $\lambda\geq0$ implies $g_\lambda > 0$ for all $\lambda\geq 0$.
 \end{prop}

It is self-evident that if $\E[H^{(t)}_T|\calF_t]=H^{(t)}_t$, then we have $\E[ G(a) ]\leq 0$
and therefore absence of statistical arbitrage. 
In fact, we  prove in the Appendix on page~\pageref{proof:nsa1} the even stronger statement:

\begin{theorem}
\label{thm:nsa1}
		The market is free from statistical arbitrage if and only if the following two conditons hold:
	\begin{enumerate}
	\item The marginal purchase price of any instrument exceeds its expected gains:
	\eq{N3buy}
		\underbrace{ H^{(t,i)}_t +  \gmup^{(i)}_t }_{\mbox{Purchase price (Ask)}} \geq \underbrace{ \E\big[H^{(t,i)}_T \big| \mathcal{F}_t\big]}_{\mbox{Expected gains}}
	\eqend
	\item The marginal sale proceeds of any instruments do not exceed the expected liability arising from the sale:
	\eq{N3sell}
		\underbrace{ H^{(t,i)}_t - \gmdn^{(i)}_t}_{\mbox{Sale proceeds (Bid)}} \leq \underbrace{\E\big[H^{(t,i)}_T \big| \mathcal{F}_t\big]}_{\mbox{Expected liability}} \ .
	\eqend
\end{enumerate}
	In particular, in the absence of transaction costs or trading constraints for the $i$th asset we recover the classic martingale condition
	\eqx
		\E\big[H^{(t,i)}_T \big| \mathcal{F}_t\big] = H^{(t,i)}_t.
    \eqend
\end{theorem}
\begin{remark}
Under the conditions of the above theorem the conditional expectation $\E\big[H^{(t,i)}_T \big| \mathcal{F}_t\big]$ defines a martingale ``micro-price'' \cite{STOIKOV} within the bid--ask spread
in the sense that
\[
\underbrace{ H^{(t,i)}_t - \gmdn^{(i)}_t  }_{\text{bid}} \leq \E\big[H^{(t,i)}_T \big| \mathcal{F}_t\big] \leq \underbrace{ H^{(t,i)}_t +  \gmup^{(i)}_t  }_{\text{ask}} \ .
\]
\end{remark}

Theorem~\ref{thm:nsa1} motivates the desire to ``remove the drift" in order to simulate market dynamics free from statistical arbitrage. To do so, we draw on the theory of minimax measures (see for example \citet{GOLL}) to construct a suitable measure with the following result. We provide a brief proof in the Appendix on page~\pageref{proof:robustmemm}
in a less restrictive setting.

\begin{theorem}[Robustly removing the Drift under Transaction Costs and Trading Constraints ]\label{th:robustmemm}
Assume that the market is constrained and that transaction costs $c_t$ are super-additive. Suppose that $a^* \in \calA$ is a (not neccessarily unique) policy that minimizes
	\eq{minutil}
		\E_\P\left[ e^{-G(a)} \right] \ ,
	\eqend
satisfying $\E_\P[e^{-G(a^*)}]>0.$

Then, the
market under the measure $\Q^*$ given by 
\eq{memm}
		\frac{ d\Q^* }{ d\P } = \frac{ e^{ -G(a^*) } }{ \E_{\P} [ e^{- G(a^*)  } ] } \ 
	\eqend

 is free from statistical arbitrage for any transaction cost $c'_t \geq c_t$ and any tighter constraints $\calA_t'\subseteq \calA_t$.
\end{theorem}

Evidently, the market under $\Q^*$ is free of statistical arbitrage for all families of monetary utility functions $U_\lambda$ with  $\E[X] = U_0(X) \geq U_\lambda(X)$.\footnote{More generally, if~a convex risk-measure $-U$ is law-invariant under $\Q^*$, then $U(X) \leq \E_{\Q^*}[X]$, c.f.~ \citet[Corollary~5.1]{CVXRM}. 
}  

Note that for generalized transaction costs, the measure~$\Q^*$ constructed here need not be a martingale measure, but one under which the drift of all tradable instruments is dominated by transaction costs, in the sense of  Theorem~\ref{thm:nsa1}. Therefore there are no statistical arbitrage strategies in the market at this level of transaction cost, or higher levels, since for any given policy, increasing transaction costs leads to lower gain. 
Indeed, for any higher transaction cost $c'_t > c_t$, the unique optimal policy at all risk aversion levels (including the risk neutral trader) under $\Q^*$ is $a=0$. To see this, note that clearly $U^*_\lambda(G(0)) = 0$ for all $\lambda \geq 0$, and for any $a \neq 0$ we can write $a_t \cdot (H_T - H_t) - c'_t(a_t) = a_t \cdot (H_T - H_t) - c_t(a_t) + c_t(a_t) - c'_t(a_t)$ and so have $\E_{\Q*}[G(a)] \leq \E_{\Q*}[\sum_t c_t(a_t) - c'_t(a_t)] < 0$ which in turn implies that  $U^*_\lambda(G(a)) < 0$ for all $\lambda \geq 0$.

The practical application of the above theorem is that we may apply a measure change through a search for a statistical arbitrage strategy with smaller but not zero transaction cost. The reason for doing so is that the inclusion of some (proportional) 
trading cost will act as an $L^1$-regularizer for the search 
of~$a^*$. Using smaller transaction costs than present in the market ensures that the resulting measure is risk-neutral even in the presence of
numerical inaccuracies. 

In the case of zero transaction costs, the resulting measure is in fact an equivalent martingale measure, and furthermore, under our assumptions that $a \in L^\infty$ and $H_t \in L^\infty$, the result coincides with the following classic result, see \citet{FRITTELLI}. 

\begin{prop}\label{prop:memm}
	Assume that generalized transaction costs are zero $c \equiv 0$. Let $a^*$ be a minimizer of \eqref{minutil}. This is equivalent to saying $a^*/\lambda$ maximizes $U_\lambda(G(a))$ for any $\lambda\in(0,\infty)$. Further assume that under these assumptions the minimizer satisfies $\E_\P[e^{-G(a^*)}]>0.$
	
	Then the measure $\Q^*$ given by the density \eqref{memm} is a martingale measure. 
		
	More specifically, the measure $\Q^*$ is the minimal entropy martingale measure (MEMM) in the sense that 
	it minimizes the relative entropy
	\eqx
		H(\Q | \P) =  \E_\Q\left[ \log \frac{d\Q}{d\P} \right] \ 
	\eqend
over all equivalent martingale measures $\Q$. Trivially, this implies that the market with trading constraints and transaction costs is free of statistical arbitrage under this measure. 
\end{prop}

To illustrate our results, it is helpful to work out the following toy example:

\begin{example}[One-period binomial model, with transaction costs]

Let $n=1$ and let $H_1 - H_0$ be a one-dimensional random variable with
\begin{equation*}
\mathbb{P}(H_1 - H_0 = u) = p, \quad \mathbb{P}(H_1 - H_0 = d) = 1-p,
\end{equation*}
where $u > d$ and $p \in (0,1)$ are parameters. Assume also that $\mathcal{F}_0 = \{\emptyset,\Omega\}$, $\mathcal{F}_1 = \sigma\{ H_1\}$. (We note that $H_0$ is then non-random.) Any admissible policy is then of the singleton form $a \in \R$, a (non-random) parameter --- let us assume here that it is unrestricted. With symmetric proportional transaction cost $\gamma > 0$ we have:
\begin{equation*}
\mathbb{P}(G(a) = a u - |a| \gamma) = p, \quad \mathbb{P}(G(a) = a d  - |a| \gamma) = 1-p.
\end{equation*}
\end{example}

Firstly,
\begin{equation*}
U_0(G(a)) = \E[G(a)] = ( a u - |a| \gamma) p + ( a d  - |a| \gamma ) (1-p) = a (u p + d (1-p)) - |a| \gamma.
\end{equation*}
Thus, $\sup_{a} U_0(G(a)) = \sup_{a \in \R} a (u p + d (1-p)) - |a| \gamma$ equals zero in the case $| u p + d (1-p)| \leq \gamma$ and $\infty$ otherwise. Note that this includes the zero transaction cost case implying then that $u p + d (1-p)=0$.

Secondly, for $\lambda \in (0,\infty)$,
\begin{equation*}
U_\lambda(G(a)) = -\frac{1}{\lambda} \log \E[e^{-\lambda G(a)}] = -\frac{1}{\lambda} \log\big( e^{-\lambda (a u - |a|\gamma)} p +  e^{-\lambda (a d- |a|\gamma)} (1-p) \big),
\end{equation*}
whereby finding the maximizer of $a \mapsto U_\lambda(G(a))$ boils down to finding the minimizer of $a \mapsto e^{-\lambda (a u- |a|\gamma)} p +  e^{-\lambda (a d- |a|\gamma)} (1-p) =: f(a)$. If $d < u \leq -\gamma$ or $u > d \geq \gamma$ (the case where there is classical arbitrage) then $f$ is strictly monotonic with $\inf_{a \in \R} f(a) = 0$, so that $\sup_a U_\lambda (G(a))= \infty$. Now, assume we are outside of those cases. For $a > 0$ we have, 

\begin{equation*}
f'(a) = -\lambda \big((u - \gamma) p e^{-\lambda (a u - a\gamma)} + (d - \gamma) (1-p)e^{-\lambda (a d - a\gamma)}\big) = 0 \quad \Leftrightarrow \quad a = \frac{\log \frac{p(u-\gamma)}{-(1-p)(d-\gamma)} }{\lambda (u-d)} =: a^*.
\end{equation*}
Now $a^*$ is ensured to exist provided that  
\begin{equation*}
 \frac{p(u-\gamma)}{-(1-p)(d-\gamma)}  > 0  \Leftrightarrow u > \gamma > d.
\end{equation*}
and it is positive (as assumed) provided that 
\begin{equation*}
 \frac{p(u-\gamma)}{-(1-p)(d-\gamma)}  > 1  \Leftrightarrow up + d(1-p) > \gamma.
\end{equation*}

By symmetry, for $a < 0$ we find that for 
\begin{equation*}
a^* := \frac{\log \frac{p(u+\gamma)}{-(1-p)(d+\gamma)} }{\lambda (u-d)}
\end{equation*}
we have $f'(a^*)=0$, which exists provided that $u > -\gamma > d$ and is negative as required provided that
\begin{equation*}
 \frac{p(u+\gamma)}{-(1-p)(d+\gamma)}  < 1  \Leftrightarrow up + d(1-p) < -\gamma.
\end{equation*}

Since $f''(a) > 0$, $f$ is convex and continuous, including at zero, where $f(0) = 1$, it must hold that $f(a^*) < 1$ in both cases. That is, we have statistical arbitrage $\sup_{a} U_\lambda(G(a)) > 0$ provided that $|up + d(1-p)| > \gamma$. Otherwise, we have no turning points for $f$ and instead acheive a global minimum at $a=0$. Note that the case $ -\gamma < d < u < \gamma$ immediately implies that $|up + d(1-p)| < \gamma$ and hence we have no statistical arbitrage.

Thirdly, for $\lambda = \infty$, 
\begin{equation*}
U_\infty(G(a)) = \essinf G(a) = \begin{cases}
a d - a \gamma, & a >0, \\
0, & a = 0, \\
a u + a \gamma, & a <0.
\end{cases}
\end{equation*}
In the so-called strong arbitrage cases $u>d > \gamma$ and $-\gamma > u > d$, choosing $a > 0$ and $a < 0$, respectively, we get $U_\infty(G(a)) >0$ and letting $|a| \rightarrow \infty$ shows that $\sup_a U_\infty(G(a)) = \infty$. In the cases $u>d = \gamma$ and $-\gamma = u > d$ we have classical (non-strong) arbitrage and $U_\infty(G(a)) \leq 0$ for any policy $a$, whereby $\sup_a U_\infty(G(a)) = 0$, attained at $a = 0$. Finally, in the arbitrage-free case we have similarly $U_\infty(G(a)) \leq 0$, so that $\sup_a U_\infty(G(a)) = 0$.





\section{Numerical Results}
To illustrate the change of measure constructed in Theorem~\ref{th:robustmemm} and Proposition~\ref{prop:memm} and, we begin by implementing the approach in some model examples where the theoretical baselines are tractable, although we stress that the approach described above does not rely on any model specification and is fully versatile to be utilized with any market simulator, in particular ``black box'' neural network simulators.

In the numerical implementation, we consider a scenario where we want to remove statistical arbitrage from a specific set of $N$ paths generated from $\P$. We parametrize our policy by a neural network $a_t = a_t(\theta_t | S_0, \ldots, S_t)$ where $\theta = (\theta_t)_t$ is the entire parameter vector. We train the network using the usual stochastic gradient descent methods, to obtain a policy $a(\theta^*)$.  We can then obtain probability weights $q^*$ for each path under $\Q^*$ via 
\eqx
q^* = \frac{e^{-G(a(\theta^*))}}{\sum e^{-G(a(\theta^*))}}
\eqend
and then expectations under $\Q^*$ are just weighted sums using these weights. 

Although neural networks are well known to be universal approximators, due to the estimation and approximation error inherent in this method, given the non-smoothness of our cost function we still expect some numerical noise from the minimizer of the $L^1$ metric, the resulting measure may not be entirely free of statistical arbitrage in the sense of Therorem ~\ref{th:robustmemm}. However, provided our trained policy is close to optimal on the second set of sample paths, we can bound the maximum utility of a risk averse trader under $\Q^*$ in the following way.

\begin{prop}
\label{prop:approximationerr}
Suppose that $a^*$ is the truly optimal policy, and let $\tilde a:= a(\theta^*)$ be the approximated policy from a neural network
with unknown approximation error $\epsilon>0$ such that $\E_\P[e^{-G(a(\theta^*))}] \leq ( 1 + \epsilon) \E_\P[e^{-G(a^*)}]$.
Denote by  $ \tilde U_\lambda$  the entropy under
under the measure $\tilde \Q$ given by $\tilde a$.

For any risk aversion $\lambda$, and any policy $a$ we have 
then
$ \tilde U_\lambda(G(a)) \leq \frac{1}{\lambda} \log ( 1 + \epsilon)$ with equality attained in the absence of transaction cost.
\end{prop}

\subsection{Toy example: Black-Scholes model} 
As a simple example, consider first a market described by a discrete version of a one-dimensional Black Scholes model with spot dynamics $S_t = S_0 \exp( (\mu - 1/2\sigma^2) t + \sigma W_t)$ where $\sigma > 0$ and $W_t$ a one-dimensional Brownian motion. No options are simulated. We ensure the integrability condition of the spot process by simulating only a finite set of paths.
Clearly, if $\mu \neq 0$ we have statistical arbitrage since the $S_t$ will no longer be a martingale, indeed for $\mu > 0$, a simple `buy and hold' policy $a_0 = c > 0$, $a_t = 0$ for all $t>0$ will produce positive utility. 

For a numerical implementation of this scenario, we set a trading horizon of 30 trading days ($T=30/252$), allowing daily rebalancing of the spot. The spot drift is $\mu = 0.05$ and the volatility is $\sigma = 0.15.$ We simulate a set of $10^6$ paths as training data for our policy. 
We parametrize our policy with a neural network trained via the Adam optimizer. Specifically, we use a two layer recurrent neural network structure with 32 units in each layer
and with global learning rate of $2 \times 10^{-5}$, batch size of 256, and trained for 100 epochs. To assess the performance, we track three metrics:

\begin{enumerate}
\item \textbf{Mean square error for $\frac{d\Q^*}{d\P}$.} In the Black Scholes model, under the assumption of continuous trading we have a unique equivalent martingale measure, given by the Radon-Nikodym density 

\eqx \frac{d\Q^*}{d\P} = \exp \left( \frac{-\mu}{\sigma}W_T - \frac{\mu^2}{2\sigma^2}T \right)
\eqend
and that $W_t + \frac{\mu}{\sigma}t$ is a standard Brownian motion with respect to $\Q^*$. Hence we can calculate this density for each simulated path in the validation set and then compare to the estimates generated by our neural network, tracking the mean square error in this estimate. 

\item \textbf{Relative entropy.} From the above, it is straightforward to obtain the entropy $H(\Q^* | \P) = \frac{\mu^2}{2\sigma^2}T$. Thus by tracking the finite sample relative entropy on our validation set we can assess convergence to this value.

\item \textbf{Mean square error for vanilla option prices.} A further simple metric is to assess whether the expected payoff of vanilla options under $\Q^*$ match with the option prices derived from the Black Scholes formula. To do so we calculate call prices on a grid of relative strikes $\mathcal{K} = \{ 0.8, 0.85, \ldots, 1.2 \}$ and calculate the mean square error between the theoretical prices and the Monte Carlo prices 
\end{enumerate} 

To reduce the Monte Carlo variance in these metrics, we evaluate them every 100 gradient steps on a separate validation set of size $10^6$. Results are plotted in Figure \ref{fig:bs_results}. We can clealy see good convergence in all three metrics, with the mean square errors of the density and option prices converging to zero, and the relative entropy converging to the correct value. 

\begin{figure}[h]
  \centering
  \includegraphics[width=1.0\textwidth]{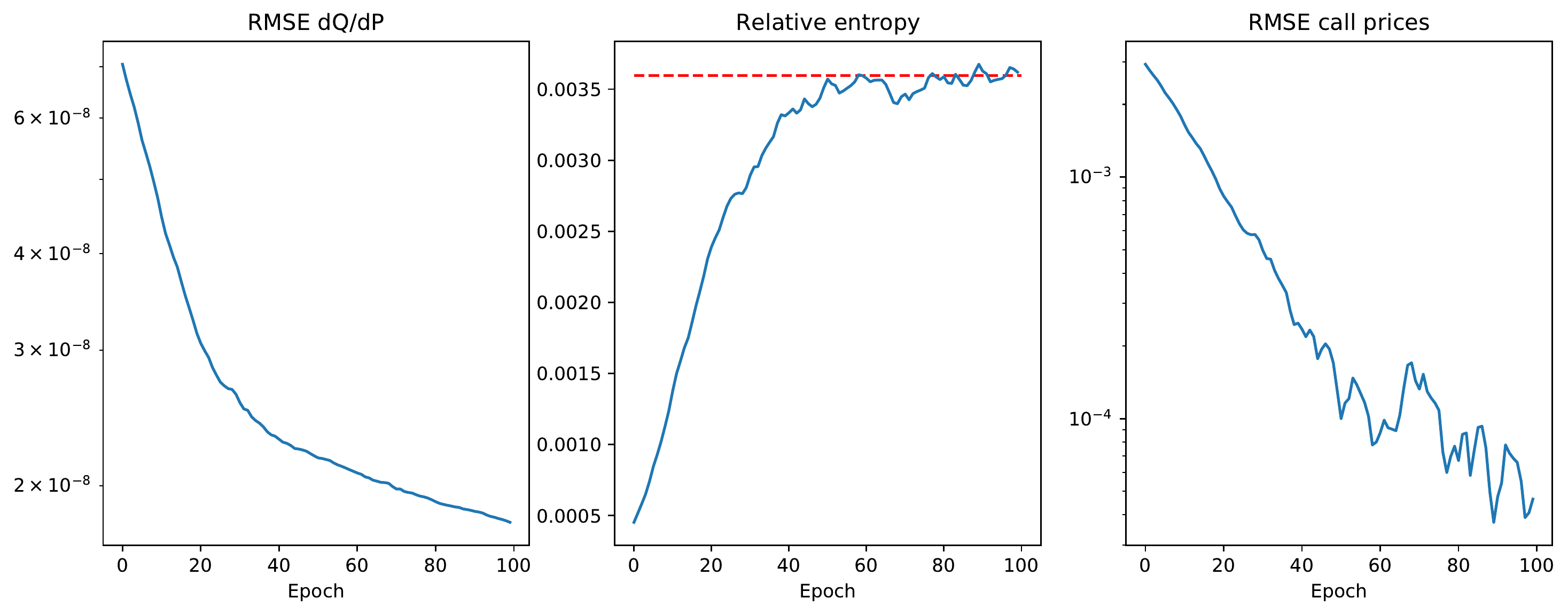}
  \caption{Metrics for the Black Scholes model. The red dashed line indicates the theoretical relative entropy of the MEMM. }
  \label{fig:bs_results}
\end{figure}

\subsection{Toy example: Black-Scholes model realized volatility} 

To illustrate the effect of the measure change, we apply it to a Black Scholes market where  spot and at the money puts and calls can be traded at each timestep.
 We simulate $10^5$ market paths where the spot is free of drift, but the options are priced with $\sigma^{implied} > \sigma^{realized}$. Specifically, we simulate market paths with $\sigma^{implied} = 0.2, \sigma^{realized} = 0.15.$
    In this case, an effective statistical arbitrage strategy would be to sell puts and calls, and delta hedge against that with
    spot.

 In discrete time, 
 equivalent measure changes can change realized volatility, contrary to the invariance of quadratic variation in continuous time. 
Hence, to remove statistical arbitrage we need to reweight the spot distribution so that the realised volatility is in line with the implied volatility. The outcome of this is demonstrated in Figure \ref{fig:realised_implied_weights} where on the left, paths with low realised volatility are given low weight, and on the right paths with high realized volatility are given high weight. 

\begin{figure}[h]
  \centering
  \includegraphics[width=0.8\textwidth]{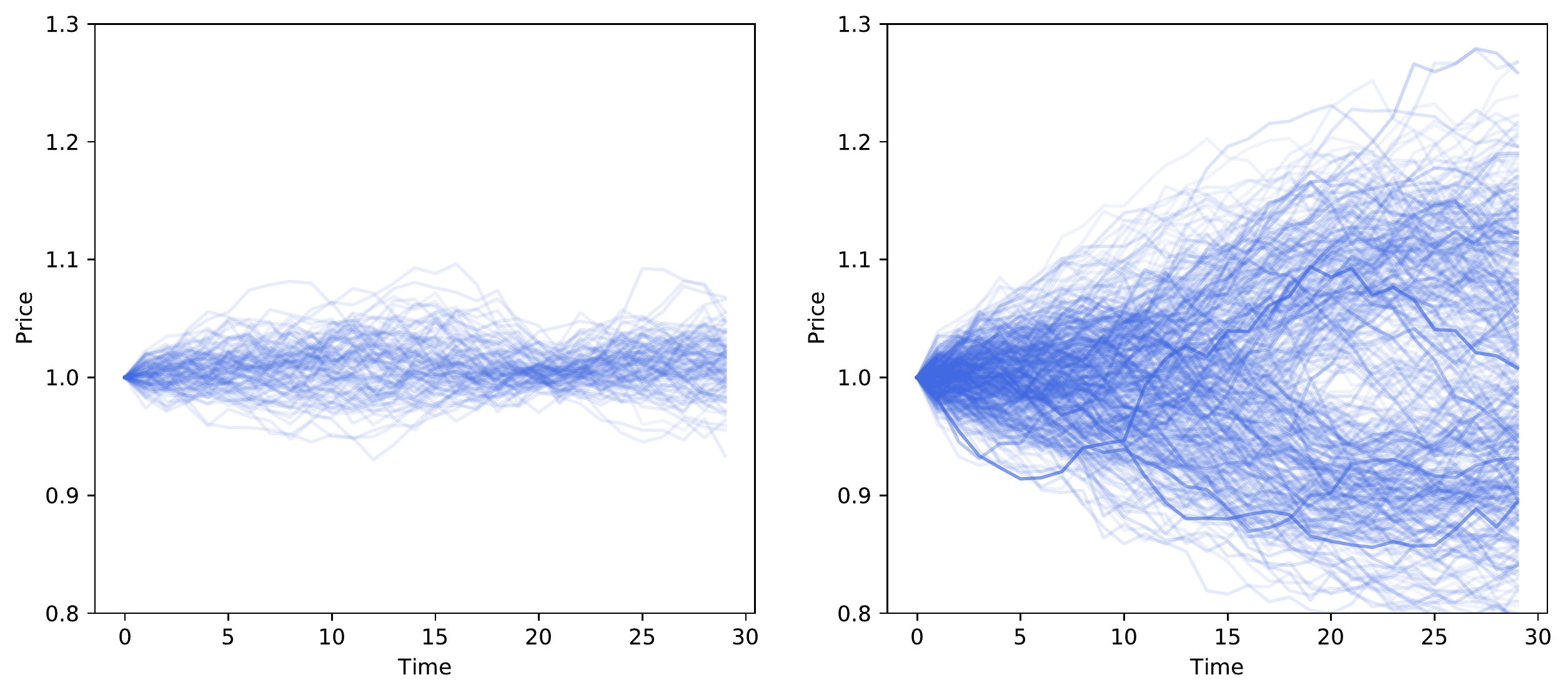}
  \caption{Paths given low weight (lowest $0.1\%$, left) and high weight (highest $0.1\%$, right) under the measure change in a Black Scholes market simulator with implied volatility higher than realized volatility.}
  \label{fig:realised_implied_weights}
\end{figure}

\subsection{Example: Vector Autoregressive model}

A simple multivariate model for spot and option prices could be to simulate a market through a vector autoregressive model.
To this end, we recap briefly the notion of \emph{discrete local volatilities}~\citet{DLV}. 
Assume maturities $0<\tau_1<\cdots<\tau_m$ and relative strikes $0<x_1<\ldots<1<\ldots<x_n$.\footnote{See~\citet{DLV} for the use of inhomogeneous strike grids.} 
We also define the additonal ghost strikes $x_0:=0$ and $x_{n_1}:=1+2 x_n \gg x_n$ for which we assume each option has intrinsic value.
Set $\tau_0:=0$. 
We ignore discounting and forwards here, but adding them is a minor extension.

For $i=1,\ldots,n$ and $j=1,\ldots,m$ we denote by $C^{j,i}$ the option with payoff $(S_{\tau_j}/S_0 - x_i)^+$ at maturity $\tau_j$.
Define
\eqx
	\Delta^{j,i} := \frac{ C^{j,i+1} - C^{j,i} }{ x_{i+1} - x_i } \ , \ \ \ \Gamma^{j,i} := \Delta^{j,i} - \Delta^{j,{i-1}}
	 \ \ \ \mbox{and} \ \ \ 
	 \Theta^{j,i} := \frac{ C^{j,i} - C^{j,i-1} }{ \tau_j - \tau_{j-1} }\ .
\eqend
The \emph{discrete local volatility} surface $(\sigma^{j,i})_{j,i}$ is given by

\eq{sigmadef}
	\sigma^{j,i} := \left\{ 
			\begin{array}{ll} 	
				\infty     & \mbox{if $\Gamma^{j,i} < 0$, or $\Theta^{j,i} < 0$, or $\Gamma^{j,i} = 0$ and $\Theta^{j,i} > 0$;}\\
							\sqrt{ \frac{2\, \Theta^{j,i} }{  x^{j,i}{}^2 \Gamma^{j,i} } } & \mbox{else.}
			\end{array}
\right.
\eqend

We recall that the options are free of static arbitrage \footnote{E.g.~there is a martingale process which generates these option prices.} if and only if $\sigma<\infty$, c.f.~\citet{DLV}.
Moreover, given a surface of finite discrete local volatilities, we may reconstruct the original surface by solving for the call prices using the implicit finite difference scheme implied by~\eqref{sigmadef}.
This involves inverting sequentially $m$ tridiagonal matrices. We note that this operation is available ``on graph" in modern automatic adjoint differentiation (AAD)
machine learning packages such as TensorFlow.
Figure~\ref{fig:histSPXdlv} shows such historic discrete local volatility surfaces as illustration.

\begin{figure}[h]
  \centering
  \includegraphics[width=0.8\textwidth]{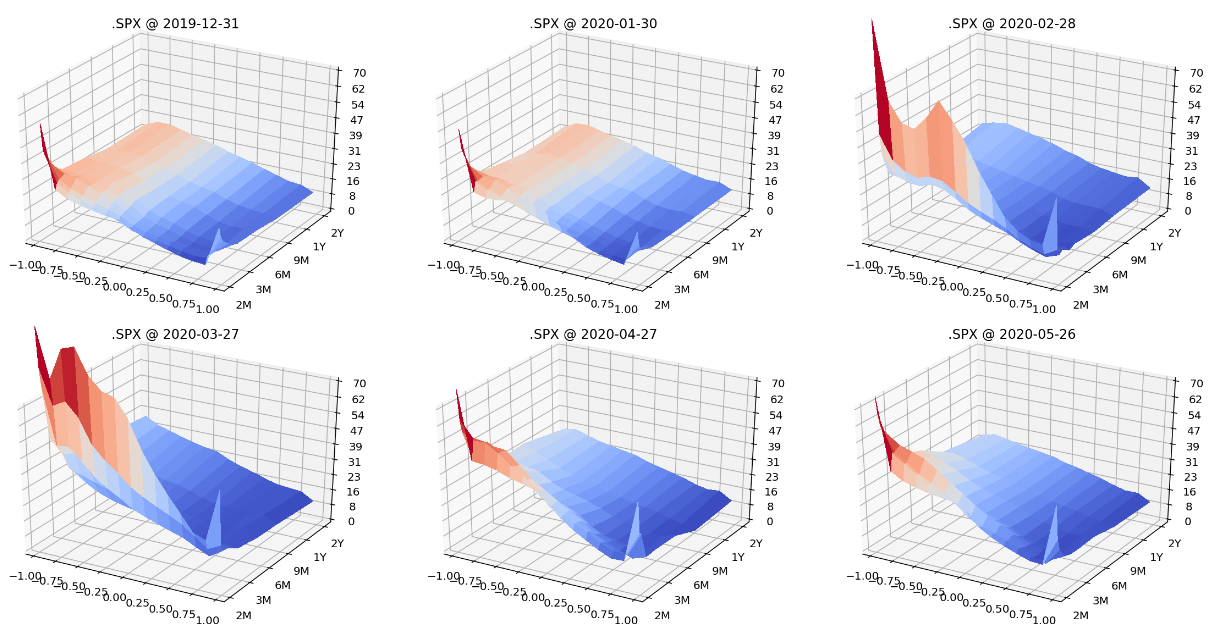}
  \caption{Historic discrete local volatility surfaces for S\&P500, in delta strikes}
  \label{fig:histSPXdlv}
\end{figure}

 \subsubsection*{Removing the Drift}
   
Given vectors $Y_t$ of (backward) log spot returns and logs of discrete local volatilities on maturities $0<\tau_1<\cdots<\tau_m$ and relative strikes $0<x_1<\ldots<1<\ldots<x_n$
we can simulate a $VAR(p)$ process:
\eqx
Y_{t} = A_{1}Y_{t-1} + ... + A_{p}Y_{t-p} + u_{t} \ , \qquad u_{t} \sim N(0, \Sigma_{u}) \ ,
\eqend
where $A_{i}$ is a $mn+1 \times mn+1$ coefficient matrix. We fit the model to data from EURO STOXX 50, using standard regression techniques from the Statsmodels Python package \cite{STATSMODELS}. Once the model has been fit, we can simulate new sample paths of log spot returns and discrete local volatilities and convert them to option prices using the methods detailed above, so that we can simulate market states of spot and option prices. To this end, we generate $10^5$ paths, of length 30 days, where each path consists of spot and both put and call option prices on a grid of maturities of $ \{ 20, 40, 60 \}$ and relative strikes $\mathcal{K} = \{ 0.85, \ldots, 1.15 \}$. For the measure change, we set transaction costs for all instruments to be proportional at level $\gamma = 0.001$.  

We then construct the measure $\Q^*$ which is free from statistical arbitrage. To parametrize the policy action, we use a two layer feedforward neural network, with 64 units in each layer and ReLU activation functions. We train for 2000 epochs on a training set of $10^5$ paths. Figure \ref{fig:realised_payoffl} compares the expected value of payoffs vs.~their prices under both the statistical and the risk-free measure in relation to trading cost. 
The expected payoff under the changed measure has clearly been flattened towards zero, and now lies within the transaction cost level, so that the drift has been removed. To confirm that statistical arbitrage has indeed been eliminated from the market simulator under this measure, we train two new, identical, network to find an optimal policy under the exponential utility, on the same simulated paths, one with unweighted, and one weighted by the output of the measure change, this time with proportional transaction costs at level $\gamma = 0.002$. Figure \ref{fig:var_pnl} shows the distributions of terminal gains of respective estimated optimal policies under $\P$ and $\Q$. Here the distribution of gains has been shifted so that it is centred at zero confirming that statistical arbitrage has been removed.

\begin{figure}[h]
  \centering
\includegraphics[width=0.8\textwidth]{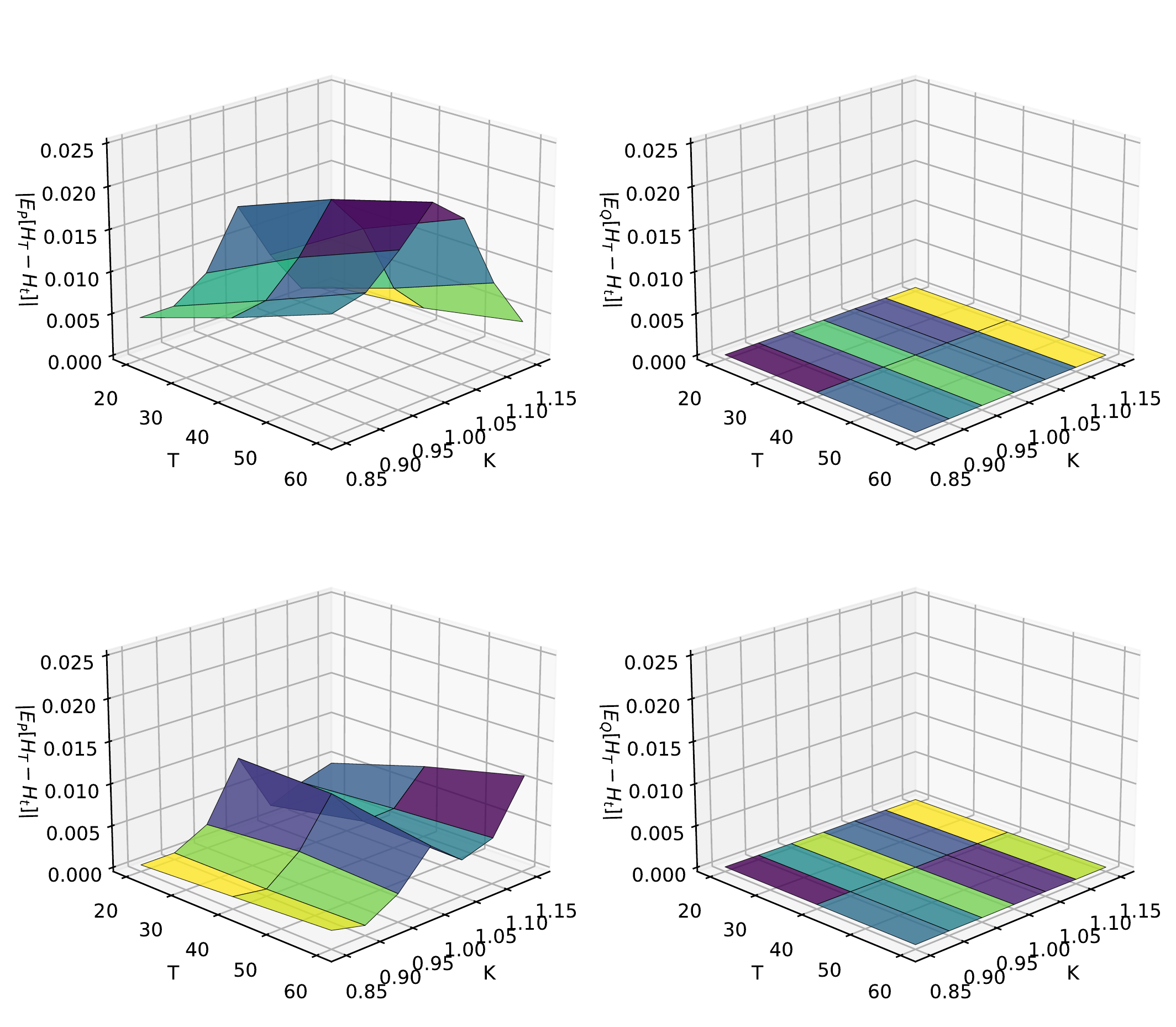}
  \caption{Average realised drift for call options (top) and put options (bottom) under the $\P$ market simulator (left) and $\Q^*$ simulator (right), by strike and maturity.}
  \label{fig:realised_payoffl}
\end{figure}

\begin{figure}[h]
  \centering
  \includegraphics[width=0.8\textwidth]{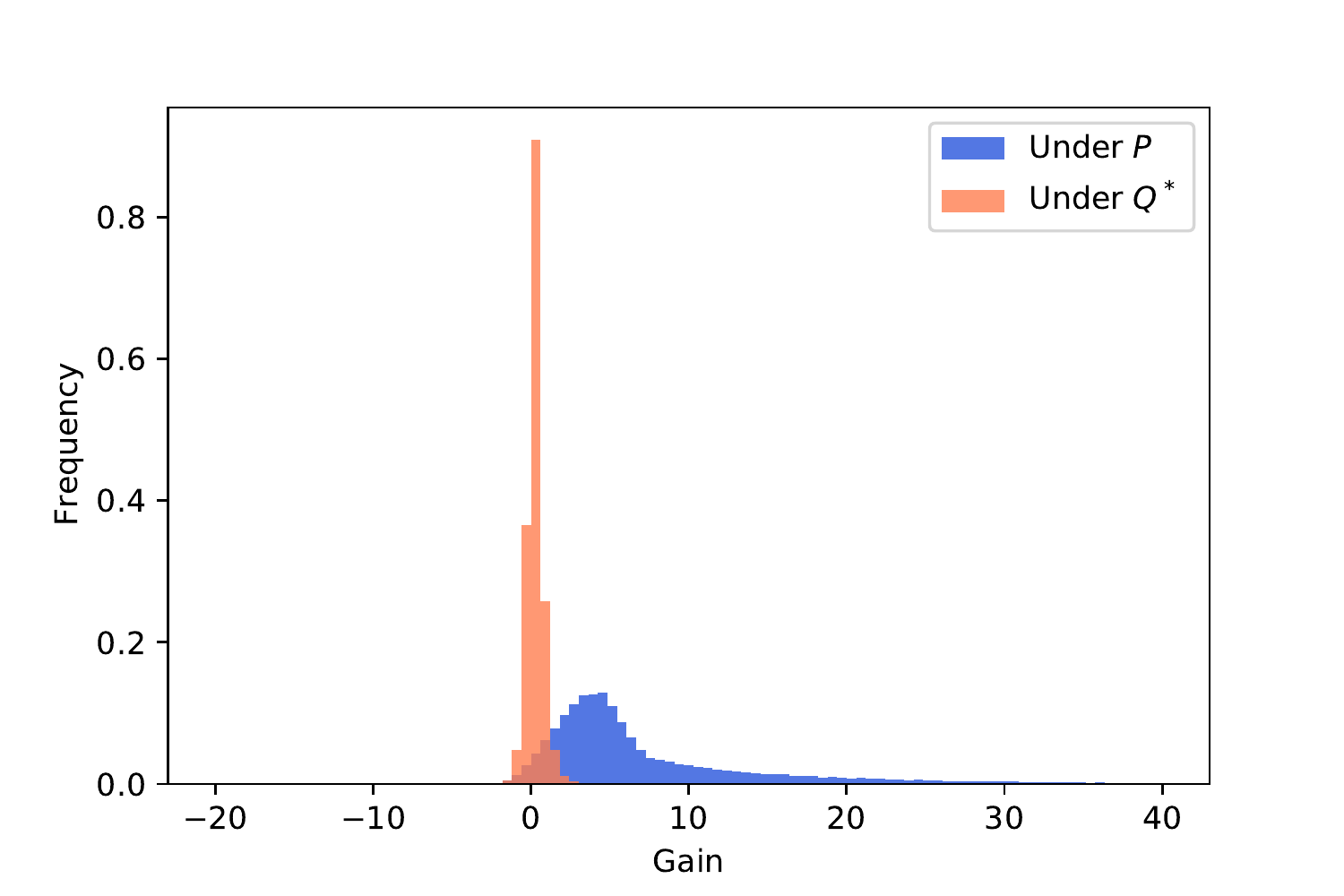}
  \caption{Gains distribution of estimated optimal policy under $\P$ and $\Q^*$}
  \label{fig:var_pnl}
\end{figure}

\section{Deep Hedging under Risk-Neutral Dynamics}

In this short section we briefly comment on the implications of using the MEMM when solving the Deep Hedging problem~\cite{DH}.

Assume here that we have a portfolio of financial instruments with terminal payoff $Z$. The \emph{Deep Hedging problem} is then 
\eq{dh}
	g_\lambda(Z) := \sup_a\ U_\lambda\big(\, Z + G(a) \,\big) \ .
\eqend
If we wanted to sell a new product with terminal payoff $X$ to a client, then our new terminal portfolio becomes $Z - X$.
The minimal risk-adjusted price is then given by
\eqx
	\pi_\lambda(X|Z) := g_\lambda(Z) - g_\lambda(Z - X) \ .
\eqend
We note that this implies our fair mark-to-market of closing our portfolio is $\pi_\lambda(Z|Z) = g_\lambda(Z) - g_\lambda$, giving
rise to the need assessing the presence of statistical arbitrage strategies.  To this end, define the 
entropy under our MEMM $\Q^*$ as,
\eqx
	U_\lambda^*( X ) := -\frac1\lambda \log \EX[\E_{\Q^*}]{ e^{-\lambda X} } \ .
\eqend
Define also the Deep Hedging problem under $\Q^*$,
\eqx
	g_\lambda^*( Z ) := \sup_a U_\lambda^*\big(\, Z + G(a)\, \big) \ . 
\eqend
Evidently, $\pi^*_\lambda(Z|Z) = g_\lambda^*( Z )$.
The following result generalizes lemma~3.3 in \citet{DH} where a similar formula was shown for the case where $Z$ has a replication strategy. The proof for the proposition and the following corollary are provided in the Appendix.

\begin{prop}\label{prop:dh1} Assume that transaction cost are super-additive i.e.~$c_t(a)+c_t(b)\geq c_t(a+b)$. Then,
\eqx
	\pi^*_\lambda(Z|Z)  = g_\lambda^*( Z ) \leq g_\lambda( Z ) - g_\lambda = \pi_\lambda(Z|Z) 
\eqend
with equality if transaction costs are zero and the problem is unconstrained.
\end{prop}

\begin{corollary}[Optimal policy for Risk Neutral Deep Hedging]\label{cor:dh2}
Assume transaction costs are zero. Under the statistical measure $\P$ suppose that $a'$ is a solution to the Deep Hedging problem for $Z$, and that $a^*$ 
is an optimal statistical arbitrage policy, i.e.~$U_\lambda(G(a^*)) = g_\lambda$.

Then the policy $a = a' - a^*$ is a solution to the Deep Hedging problem under the minimal entropy martingale measure $\Q^*$.
\end{corollary}

The previous two results show that in the absence of transaction costs, solving the Deep Hedging problem under $\Q^*$ directly removes the statistical arbitrage element of the policy that was present under $\P$. Indeed, the risk-neutral Deep Hedging problem could be solved implicitly by solving the two optimisation problems under $\P$ and taking their difference.

\section{Conclusion}

We have presented a numerically efficient method for computing a risk-neutral density for a set of paths over a number of time steps. Our method is applicable to paths of derivatives and option prices
in particular, hence we effectively provide a framework for statistically learned \emph{stochastic implied volatility} models using only basic linear algebra and comonly available machine learning tools. 
Our method is generic and does not depend on the market simulator itself, except that it requires that the simulator does not produce static arbitrage opportunities. It also caters naturally
for transaction cost and trading constraints, and is easily extended to multiple assets.

Finally, with Theorem~\ref{thm:nsa1} we have also provided novel insights into the relationship between statistical arbitrage under trading frictions and to what degree the prices of instruments may deviate from their expected values.
\appendix

\section{Proofs}

\subsection{Proof of Lemma~\ref{lem:g-prop}}

\begin{proof}[of Lemma~\ref{lem:g-prop}]
Consider $0 < \lambda \leq \lambda' < \infty$, and let $X$ be a random variable such that $\E[|X|]<\infty$. Firstly, we have
\begin{equation*}
U_{\lambda'}(X)  = - \frac{1}{\lambda'} \log \mathbb{E}[e^{-\lambda' X}] = - \frac{1}{\lambda'} \log \mathbb{E}\big[\big(e^{-\lambda X}\big)^{\frac{\lambda'}{\lambda}}\big],
\end{equation*}
where the map $x \mapsto x^{\frac{\lambda'}{\lambda}}$ is convex since $\lambda' \geq \lambda$. Thus, by Jensen's inequality,
\begin{equation*}
U_{\lambda'}(X) = - \frac{1}{\lambda'} \log \mathbb{E}\big[\big(e^{-\lambda X}\big)^{\frac{\lambda'}{\lambda}}\big] \leq - \frac{1}{\lambda'} \log \mathbb{E}\big[\big(e^{-\lambda X}\big)\big]^{\frac{\lambda'}{\lambda}} = - \frac{1}{\lambda} \log \mathbb{E}[e^{-\lambda X}] = U_\lambda(X),
\end{equation*}
(We may have $U_{\lambda'}(X) = -\infty$ or $U_{\lambda'}(X) = -\infty = U_\lambda(X)$, but the inequality remains nevertheless valid.) Secondly, since $x \mapsto e^{-\lambda x}$ is also convex, Jensen's inequality further implies $\mathbb{E}[e^{-\lambda X}] \geq e^{-\lambda \mathbb{E}[X]}$, whereby
\begin{equation*}
U_0(X) = \E[X] = -\frac{1}{\lambda} \log (e^{-\lambda \E[X]}) \geq -\frac{1}{\lambda} \log \mathbb{E}[e^{-\lambda X}] = U_\lambda(X).
\end{equation*}
Thirdly, since $e^{-\lambda' \essinf X} \geq \E[e^{-\lambda' X}]$, we have
\begin{equation*}
U_{\infty}(X) = \essinf X = - \frac{1}{\lambda'} \log e^{-\lambda' \essinf X} \leq  - \frac{1}{\lambda'} \log \E[e^{-\lambda' X}] = U_{\lambda'}(X).
\end{equation*}
In summary,
\begin{equation}\label{eq:U-ineq}
U_0(X) \geq U_{\lambda}(X) \geq U_{\lambda'}(X) \geq U_\infty(X).
\end{equation}

Now, for $0 \leq \lambda \leq \lambda' \leq \infty$, 
\begin{equation*}
g_{\lambda} = \sup_{a} U_\lambda (G(a)) \geq \sup_{a} U_{\lambda'} (G(a)) = g_{\lambda'},
\end{equation*}
by \eqref{U-ineq}, which proves that $\lambda \mapsto g_\lambda$ is non-increasing. To establish non-negativity, it remains to note that
\begin{equation*}
g_\infty = \sup_{a} U_\infty (G(a)) \geq U_\infty (G(0)) = U_\infty (0) = \essinf 0 = 0.
\end{equation*}
\end{proof}

\subsection{Proof of Proposition \ref{prop:alllambda}}

\begin{proof}[of Proposition \ref{prop:alllambda}]\label{proof:alllambda}
we will prove that $g_{\lambda} = 0$ for all $\lambda \in (0,\infty)$ implies $g_0 = 0$ for the entropy, since then the result follows from  monotonicity of $g_\lambda$. Suppose instead that $g_0 > 0$, while $g_\lambda = 0$ for all $\lambda >0$. Then there is an admissible policy $a$ such that
\begin{equation*}
U_0(G(a)) = \E[G(a)] > 0.
\end{equation*}
Since $H \in L^\infty$, $a \in L^\infty$, for this policy $G(a)$ is almost surely bounded, and so we have
\begin{equation*}
U_0(G(a)) = \lim_{\lambda \rightarrow 0+} U_\lambda (G(a))
\end{equation*}
by \citet{FOLLMER}, implying in turn that $U_\lambda (G(a)) > 0$ for some $\lambda  > 0$. But this contradicts the assumption $g_\lambda = 0$ for all $\lambda >0$, so it follows that $g_0 = 0$.
\end{proof}

\subsection{Proof of Theorem \ref{thm:nsa1}}\label{proof:nsa1}
To prove Theorem \ref{thm:nsa1}, we need a few auxiliary results:

\begin{lemma}\label{lem:cond-exp}
Let $X$ be a random variable such that $\E[|X|]<\infty$ and let $\mathcal{G} \subset \mathcal{F}$ be a $\sigma$-algebra. Suppose that $Y$ is a non-negative, $\mathcal{G}$-measurable random variable such that 
\begin{equation}\label{eq:exp-bound}
|\E[1_A X ]| \leq \E[1_A Y] <\infty \quad \text{for any $A \in \mathcal{G}$.}
\end{equation}
Then,
\begin{equation}\label{eq:cond-exp-bound}
|\E[X | \mathcal{G}]| \leq Y \quad \text{$\prob$-a.s.}
\end{equation}
\end{lemma}

\begin{proof}
we note that if \eqref{cond-exp-bound} does not hold, then we have $\P[\E[X | \mathcal{G}] > Y]>0$ or $\P[\E[X | \mathcal{G}] < -Y]>0$.
In the former case, let $A : = \{ \E[X | \mathcal{G}] > Y\} 
$ 
so that $A \in \mathcal{G}$ and $\P[A]>0$. We have then
\begin{equation*}
\E[1_A X] = \E[1_A \E[X|\mathcal{G}]]  
> \E[1_A Y],
\end{equation*}
which contradicts the assumption \eqref{exp-bound}.
In the latter case, take similarly $A : = \{ \E[X | \mathcal{G}] < -Y\} 
$ whereby again $A \in \mathcal{G}$ and $\P[A]>0$. Now,
\begin{equation*}
\E[1_A X] = \E[1_A \E[X|\mathcal{G}]] 
< -\E[1_A Y],
\end{equation*}
contradicting \eqref{exp-bound} as well.
\end{proof}

We will aim to reduce the proof of Theorem \ref{thm:nsa1} to the case of proportional costs. To this end,
define by $\tilde c_t$ the respective proportional cost
\eq{tildeC}
	\tilde c_t(a_t) := \left\{
	\begin{array}{ll}
		\gmup_t \cdot a^+_t + \gmdn_t \cdot a^-_t , & a_t \in \calA_t , \\
		\infty \ , & a_t \not \in \calA_t \ ,
	\end{array}
	\right.
\eqend
and the associated gains process
\eqx
 	\tilde G(a) := \sum_{t=0}^{m-1} a_t \cdot ( H^{(t)}_T - H^{(t)}_t ) - \tilde c_t(a_t) \ .
\eqend

\begin{lemma}\label{lem:Gc}
We have
$\sup_a \E[ \tilde G(a) ] = 0$ if and only if $\sup_a \E[ G(a) ] = 0$.
\end{lemma}

\begin{proof}
 Since by construction $G(a)\leq \tilde G(a)$ we have to show that
 $\sup_a \E[ G(a) ]= 0$ also implies $\sup_a \E[ \tilde G(a) ]= 0 $.
Assume the contrary, i.e.~$a^*$ is a strategy
such that $\E[ \tilde G(a^*) ]=b>0$ but $\sup_a \E[ G(a) ]= 0$. 
Concavity of $G$ and $G(0)=0$ imply that $\E[ G(\varepsilon a^*) ] \geq \varepsilon b$ and therefore
$\E[ \frac1\varepsilon G(\varepsilon a^*) ] \geq  b$. Taking the monotone limit $\varepsilon\downarrow 0$ yields the contradiction
$\E[ \tilde G(a^*) ] = b > 0$.
\end{proof}

\begin{proof}[of Theorem \ref{thm:nsa1}]
Thanks to Lemma~\ref{lem:Gc} we may focus on the case where transaction costs are proportional.
We first prove that if the market is free from statistical arbitrage, then inequalities \eqref{N3buy} and \eqref{N3sell} in the theorem hold.

Define 
 \eqx
 	\adn^{(i)}_t := \min\left\{ 1, - \half \inf_{a\in\calA} a_t^{i} \right\} \ \ \ \mbox{and} \ \ \ \aup^{(i)}_t := \min\left\{1,\half \sup_{a\in\calA} a_t^{i} \right\}\ .
 \eqend
Evidently, $\adn^{(i)}_t, \aup^{(i)}_t \in[0,1]$ and, in particular, $\adn^{(i)}_t > 0$ whereever $\gmdn^{i}_t < \infty$ and $\aup^{(i)}_t > 0$ where $\gmup^{i}_t < \infty$.

Fix $i$ and $t$ where $\aup^{(i)}_t>0$ is not empty, let $A\in\calF_t$ be arbitrary, and let $\overline{a}$ be the policy which is zero except over $A$ where we buy $\aup^{(i)}_t$
units of $H^{(i)}_t$. Note that $\overline{a}\in\calA$.
\begin{equation*}
\begin{split}
0 \geq U_0(\tilde G(\overline{a})) = \E[\tilde G(\overline{a})] & = \E\big[1_A\aup^{(i)}_t \big(H^{(t,i)}_T - H^{(t,i)}_t\big)\big] - \E\big[\gmup^{(i)}_t \big|1_A\aup^{(i)}_t \big| \big] \\
& =\E\big[1_A \aup^{(i)}_t\big(H^{(t,i)}_T - H^{(t,i)}_t\big)\big] - \E\big[1_A \aup^{(i)}_t \gmup^{(i)}_t  \big] \\
& = \E\big[1_A\ \aup^{(i)}_t\ \left\{\   \E\big[H^{(t,i)}_T - H^{(t,i)}_t\big| \mathcal{F}_t\big] -  \gmup^{(i)}_t\ \right\} \big],
\end{split}
\end{equation*}
which implies on $\{\aup^{(i)}_t > 0\}$ with Lemma~\ref{lem:cond-exp}
\eq{N3_1}
	H^{(t,i)}_t + \gmup^{(i)}_t  \geq 
	\E\big[ H^{(t,i)}_T \big| \mathcal{F}_t\big]   
\eqend
i.e.~equation \eqref{N3buy}. Note that on $\{ \aup^{(i)}_t =  \}$ we have $\{ \gmup^{(i)}_t = \infty \}$ so that \eqref{N3buy} holds trivially there, too.

For selling, fix again $i$ and $t$ where $\adn^{(i)}_t>0$ is not empty, let $A\in\calF_t$ agan be arbitrary,  and let $\underline{a}$ be the policy which is zero except over $A$
where we sell $\adn^{(i)}_t$
units of $H^{(i)}_t$.
\begin{equation*}
\begin{split}
0 \geq U_0(\tilde G(\underline{a})) = \E[\tilde G(\underline{a})] & = \E\big[1_A\ \adn^{(i)}_t\ \left\{\  - \E\big[H^{(t,i)}_T - H^{(t,i)}_t\big| \mathcal{F}_t\big] -  \gmdn^{(i)}_t  \ \right\} \big],
\end{split}
\end{equation*}
Reording on $\{\adn^{(i)}_t > 0\}$ yields
\eq{N3_2}
	H^{(t,i)}_t - \gmdn^{(i)}_t \leq 
	\E\big[ H^{(t,i)}_T \big| \mathcal{F}_t\big]   
\eqend
and therefore equation~\eqref{N3sell}.

Let us now prove the reverse statement that if \eqref{N3_1} and \eqref{N3_2} hold, then for $g_0=0$. Recall that for any admissible $a$ we have $a_t = a^+_t - a^-_t$
with $a^+,a^-\geq 0$.
\begin{equation*}
g_0 = \E[G(a)] = \sum_{t=0}^{m-1} \sum_{i=1}^n \E\big[ (a^{(i)+}_t - a^{(i)-}_t)\big(H^{(t,i)}_T - H^{(t,i)}_t\big)- \left\{ \gmup^{(i)}_t a^{(i)+}_t + \gmdn^{(i)}_t a^{(i)-}_t \right\}  \big] \ .
\end{equation*}
We have
\eqx
\E\big[a^{(i)+}_t\big(H^{(t,i)}_T - H^{(t,i)}_t\big)-\gmup^{(i)}_t a^{(i)+}_t  \big] 
= \E\big[ a^{(i)+}_t  \big\{\  \E\big[ \big(H^{(t,i)}_T - H^{(t,i)}_t\big)\big| \mathcal{F}_t \big] -\gmup^{(i)}_t \ \big\} \big] 
 \supstack{\eqref{N3_1}}\leq 0
\eqend
and
\eqx
\E\big[- a^{(i)-}_t\big(H^{(t,i)}_T - H^{(t,i)}_t\big)-\gmdn^{(i)}_t a^{(i)-}_t  \big] 
= \E\big[ a^{(i)+}_t  \big\{\ - \E\big[ \big(H^{(t,i)}_T - H^{(t,i)}_t\big)\big| \mathcal{F}_t \big] -\gmdn^{(i)}_t \ \big\} \big] 
 \supstack{\eqref{N3_2}}\leq 0
\eqend
which shows $g_0=0$.
\end{proof}

%
%

\begin{remark}\label{lemma:zeroinf}
It may be of interest to note the following dichotomy for the entropy with proportional transaction costs, which follows since $X \mapsto U_0(X)$ is linear and $X \mapsto U_\infty(X)$ positive homogeneous. In the unconstrained case, for $\lambda \in \{0,\infty\}$ either $g_\lambda= 0$ or $g_\lambda = \infty$. To see this, note that on the one hand, if $g_\lambda \leq 0$ then $g_\lambda = 0$ by Lemma \ref{lem:g-prop}.
On the other hand, if $g_\lambda > 0$ then there exists an admissible policy $a$ such that $U_\lambda(G(a))>0$. But for any constant $c>0$ the policy $ca$  is also admissible, while for $\lambda \in \{0,\infty\}$ we have 
\begin{equation*}
U_\lambda(G(ca)) = U_\lambda(c G(a)) = c U_\lambda(G(a))
\end{equation*}
by positive homogeneity. Letting $c \rightarrow \infty$ shows that then $g_\lambda = \infty$. Note that if the policy is constrained, then the case $g_\lambda=\infty$ is replaced by an optimal strategy being a boundary point in $\mathcal{A}$.
\end{remark}

\subsection{Proof of Theorem~\ref{th:robustmemm}}

\begin{proof}[of Theorem~\ref{th:robustmemm}]\label{proof:robustmemm}
superadditivity means $c_t(a^*) + c_t( a)\geq  c_t(a^*+ a) $ and therefore
$
 -G(a^*) -G(a) \geq -G(a^* + a)
$.
For $\Q^*$ as defined as in Theorem~\ref{th:robustmemm} we find
\eqx
\inf_a\E_{\Q^*} \left[ e^{- G(a)} \right] = \frac{ \inf_a\E_{\P} \left[ e^{ -G(a^*) -G(  a)} \right]} { \E_{\P} \left[ e^{ -G(a^*)} \right]} 
\supstack{(*)}\geq
\frac{ \inf_a \E_{\P} \left[ e^{ -G(a^* + a)} \right]} { \E_{\P} \left[ e^{ -G(a^*)} \right]} 
\supstack{(**)}=
\frac{ \E_{\P} \left[ e^{ -G(a^*)} \right]} { \E_{\P} \left[ e^{ -G(a^*)} \right]}  = 1 
\eqend
where $(*)$ becomes an equality in the absence of transaction cost, and where $(**)$ is due to optimality of $a^*$.
	With the results of
	Proposition~\ref{prop:alllambda} and 
Theorem~\ref{thm:nsa1}, this implies that the market under $\Q^*$ with cost $c$ and constraints $\calA$ is free from statistical arbitrage. The extension to higher cost
or tighter restrictions is trivial.
\end{proof}

\subsection{Proof of Proposition ~\ref{prop:approximationerr}}

\begin{proof}[of Proposition~\ref{prop:approximationerr}]\label{proof:approximationerr}
This follows in a similar fashion to the proof of Theorem \ref{th:robustmemm} above:
\eq{prox}
\inf_a \E_{\tilde \Q} \left[ e^{-\lambda G(a)} \right] = \frac{ \inf_a \E_{\P} \left[ e^{ -G(a(\theta^*)) - \lambda G(a)} \right]} { \E_{\P} \left[ e^{ -G(a(\theta^*))} \right]} 
\supstack{(*)}\geq
\frac{ \inf_a \E_{\P} \left[ e^{ -G(a(\theta^*) + \lambda a)} \right]} { \E_{\P} \left[ e^{ -G(a(\theta^*))} \right]} 
\supstack{(**)}=
\frac{\E_{\P} \left[ e^{ -G(a^*)} \right]} { \E_{\P} \left[ e^{ -G(a(\theta^*))} \right]}  
\geq \frac{1}{1+\epsilon} 
\eqend
where again $(*)$ becomes an equality in the absence of transaction cost, and where $(**)$ is  due to optimality of $a^*$.
   This implies that $\tilde U_\lambda(G(a)) \leq  \frac{1}{\lambda} \log(1+\epsilon)$ for $\tilde U_\lambda$ being the entropy
under the measure $\tilde \Q$ implied by $\tilde a=a(\theta^*)$.
\end{proof}

\subsection{Proof of Proposition~\ref{prop:dh1} and Corollary  \ref{cor:dh2}}

\begin{proof}[for Proposition~\ref{prop:dh1}] For this proof, 
we note that  $\lambda G(a) \geq G(\lambda a)$ for all $\lambda\geq 0$. We have equality if trading cost are  proportional and finite (i.e., unconstrained).
Since $c_t$ is superadditive we have~$G(a) + G(b) \geq G(a+b)$.
\eqary
	g_\lambda^*( Z ) & = & 
	-\frac1\lambda \log \inf_a \E_{\Q^*} \left[ e^{-\lambda (Z + G(a)) }\right] \\
	& = &  
	-  g_\lambda -\frac1\lambda \log  \inf_a \E_\P\left[ e^{-\lambda \{Z + G(a) + G(a^*)/\lambda\} }\right]
	\\
	& \supstack{(*)}\leq &  
	-  g_\lambda -\frac1\lambda \log  \inf_a \E_\P\left[ e^{-\lambda \{Z + G(a) + G(a^*/\lambda)\} }\right]
	\\
	&
	\supstack{(**)}\leq
	&
	- g_\lambda - \frac1\lambda \log  \inf_{a'=a+a^*/\lambda} \E_\P\left[ e^{-\lambda (Z + G(a')) }\right]
	\\
	& = &
	g_\lambda(Z) - g_\lambda \ .
\eqaryend
We have equality in $(*)$ if transaction costs are proportional and in $(**)$ if transaction costs are zero.
\end{proof}

\begin{proof}[for Corollary \ref{cor:dh2}]
This follows directly from the previous result and the proof is virtually identical:
\eqx
	U_\lambda^*(Z_T + G(a' - a^*))	
	=  
	-\frac1\lambda \log \E_\P\left[ e^{-\lambda (Z_T + G(a')) }\right] - g_\lambda 
	=
	g_\lambda^\P(Z_T) - g_\lambda 
	=
	g_\lambda^*(Z_T) \ .
\eqend
\end{proof}

\bibliography{ReferencesSIV}

\bibliographystyle{icml2019}

\section*{Disclaimer}

{\small
Opinions and estimates constitute our judgement as of the date of this Material, are for informational purposes only and are subject to change without notice. It is not a research report and is not intended as such. Past performance is not indicative of future results. This Material is not the product of J.P. Morgan’s Research Department and therefore, has not been prepared in accordance with legal requirements to promote the independence of research, including but not limited to, the prohibition on the dealing ahead of the dissemination of investment research. This Material is not intended as research, a recommendation, advice, offer or solicitation for the purchase or sale of any financial product or service, or to be used in any way for evaluating the merits of participating in any transaction. Please consult your own advisors regarding legal, tax, accounting or any other aspects including suitability implications for your particular circumstances.  J.P. Morgan disclaims any responsibility or liability whatsoever for the quality, accuracy or completeness of the information herein, and for any reliance on, or use of this material in any way.  \\
\noindent
 Important disclosures at: www.jpmorgan.com/disclosures}


\end{document}